\def\ARXIV{}
\def\EXTERNALCOMPETITION{}
\def\year{2020}\relax
\documentclass[letterpaper]{article} 
\usepackage{aaai20}  
\usepackage{times}  
\usepackage{helvet} 
\usepackage{courier}  
\usepackage[hyphens]{url}  
\usepackage{graphicx} 
\usepackage{graphicx}  
\usepackage{amsmath,amsthm,amsfonts,amssymb,xcolor}
\newcommand{\range}[2]{\in\{#1,\ldots,#2\}}
\newcommand{\gft}{\textsc{GFT}}

\newcommand{\citet}[1]{\citeauthor{#1}~\shortcite{#1}}
\newcommand{\citep}{\cite}

\newcommand{\floor}[1]{\text{floor}(#1)}
\newcommand{\recipe}{\mathbf{r}}

\usepackage{amsmath,amsthm}

\newtheorem{theorem}{Theorem}
\newtheorem{lemma}{Lemma}
\newtheorem{claim}{Claim}

\usepackage{csvsimple}

\title{Strongly Budget Balanced Auctions for Multi-Sided Markets}

\author{
Rica Gonen\\
The Open University of Israel\\
ricagonen@gmail.com
\And
Erel Segal-Halevi\\
Ariel University, Ariel, Israel\\
erelsgl@gmail.com
}

\begin{document}

\maketitle

\begin{abstract}
In two-sided markets, 
Myerson and Satterthwaite's impossibility theorem states that one can not maximize the gain-from-trade while also satisfying truthfulness, individual-rationality and no deficit.  Attempts have been made to circumvent Myerson and Satterthwaite's result by attaining approximately-maximum gain-from-trade: the double-sided auctions of McAfee (1992) is truthful and has no deficit, and the one by Segal-Halevi et al. (2016) additionally has no surplus --- it is strongly-budget-balanced. 
They consider two categories of agents --- buyers and sellers, where each trade set is composed of a single buyer and a single seller.

The practical complexity of applications such as supply chain require one to look beyond two-sided markets.  Common requirements are for: buyers trading with multiple sellers of different or identical items, buyers trading with sellers through transporters and mediators, and sellers trading with multiple buyers.  We attempt to address these settings.

We generalize Segal-Halevi et al. (2016)'s strongly-budget-balanced double-sided auction setting to a multilateral market where each trade set is composed of any number of agent categories. Our generalization refines the notion of competition in multi-sided auctions by introducing the concepts of external competition and trade reduction. 
We also show an obviously-truthful implementation of our auction using multiple ascending prices.

\ifdefined\ARXIV
\emph{This is a full version of a paper presented in the AAAI 2020 conference.}
\else
\emph{Full version, including omitted proofs and simulation experiments, is available at \url{https://arxiv.org/abs/1911.08094}.}
\fi
\end{abstract}

\section{Introduction}
Mechanism design for one-sided markets has been investigated for several decades in economics and in computer science. It aims to find an efficient (high social welfare) allocation of a set of items to a set of agents, while ensuring that truthfully reporting the input data is the best strategy for the agents. The Vickrey-Clarke-Groves (VCG) auction \cite{V61,C71,G73} is a pillar of mechanism design.  VCG auctions maximize the social welfare of the agents. They are dominant-strategy truthful (DST) --- each agent's dominant strategy is to truthfully report its preferences to the auction, regardless of what the other agents report. 
They can also be made individually rational (IR) --- no agent loses from participating in the auction.

More recently, there has been increased attention on auctions for two-sided markets, in which the set of agents is partitioned into buyers and sellers. As opposed to the one-sided setting, where the auctioneer initially holds the items, in the two-sided setting the items are initially held by the set of sellers. The sellers express valuations for the items they hold, and are assumed to act rationally and strategically. Thus, the auctioneer is tasked with deciding which buyers and sellers should trade and with what prices. 

The growing interest in two-sided markets can be attributed to various important applications. Examples range from selling display-advertising on ad exchange platforms, the US FCC spectrum license reallocation, and stock exchanges. However, little work has been done so far on the next level of generalization, i.e., multi-sided markets.

In two-sided markets, a further important requirement is strong budget-balance (SBB), which states that monetary transfers happen only among the agents in the market.  This means that buyers and sellers are allowed to trade without leaving the auctioneer any share of their gains and without the auctioneer adding money into the market. A weaker version of SBB, often considered in the literature, is weak budget-balance (WBB).  WBB only requires the auctioneer not to add money to the market. 
The problem with weak budget-balance is that the surplus of the auctioneer might consume most of the gain-from-trade, leaving little gain for the actual traders. This might drive traders away from the market.%
\footnote{
The following ``trick'' can be used to convert any WBB auction to an SBB auction: 
before the auction starts, remove a random trader from the market;
after the auction ends, give that trader all the surplus (if any).
We do not support this trick since it might induce agents who have nothing to do with the auction (e.g. ``sellers'' with nothing to sell or ``buyers'' with no money) to come to the market, only because of the chance to win all the surplus. 
Like \citet{colini2017approximately}, we focus on \emph{direct-trade auctions} --- auctions that give/take money only to/from agents who actually participate in the trade.
}
Note that, in bilateral trade settings, VCG is usually not even WBB except in special cases \cite{GMAC13}.

For double-sided auctions,
the impossibility theorem of 
\citet{myerson1983efficient} states that one can not maximize gain from trade (GFT, the difference between the total value of the sold items for the buyers and the total value of these items for the sellers) while also satisfying IR, DST, and no deficit.  

The seminal double-auction mechanism of
\citet{mcafee1992dominant} is DST, IR and WBB.
It
circumvents Myerson and Satterthwaite's result by compromising on GFT:
it may remove up to one deal from the optimal trade. In case a deal is removed, it is the one with the smallest GFT among the deals in the optimal trade; hence it attains at least $1-1/k$ of the optimal GFT, where $k$ is the number of optimal deals. Thus, it is \textbf{asymptotically optimal} --- its GFT approaches the optimum when $k\to\infty$.
\ifdefined\ARXIV
\footnote{
McAfee's mechanism is WBB because, when a deal is removed, the values of the removed agents are used to set the prices for the remaining agents: all buyers pay the value of the removed buyer, and all sellers receive the value of the removed seller. Thus there is a surplus that goes to the market manager.
}
\fi

Recently, \citet{SegalHalevi2016SBBA} presented a SBB variant of McAfee's mechanism, with similar GFT guarantees. Their mechanism may remove up to one \emph{buyer} from the optimal trade, and it is the buyer with the lowest value among the buyers in the optimal trade. 
In case a buyer is removed, the remaining $k-1$ buyers trade with $k-1$ sellers selected at random from the $k$ sellers in the optimal trade.
\ifdefined\ARXIV
\footnote{
Their mechanism is SBB because, when a buyer is removed, his value is used to set the prices for all remaining agents: all $k-1$ remaining buyers pay the value of the removed buyer, and the $k-1$ sellers selected at random receive the same value. Thus there is no deficit and no surplus.
Their mechanism has a variant in which a seller may be removed, instead of a buyer.
}
\fi

The complexity of practical requirements in areas such as supply chain require one to look beyond double-sided markets.
As an example \citep{babaioff2005incentivecompatible},
a market for lemonade may contain two kinds of sellers (lemon pickers and sugar producers), two kinds of buyers (juice drinkers and lemonade drinkers), and some intermediary agents (lemon squeezers, lemonade mixers, etc.)
Our goal is to address such settings while keeping the strong budget balance requirement.

\subsection{Our Contribution}
Our contribution is twofold: First we generalize \citet{SegalHalevi2016SBBA}'s SBB double auction to a multi-sided market where the trade set is composed of any number of agent types and any number of copies of any agent's type. Our generalization refines the notion of competition in multi-sided auctions by introducing the concept of external competition --- competition over who will act as a given participant's trade partner(s) (complementarity). 

The expanded notion of competition allows us to provide a simple well-performing procedure that generalizes \cite{mcafee1992dominant}'s trade reduction. The shift to thinking in terms of competition allows us to broadly address situations common to multi-sided auctions. These settings include trading entities that may be individuals or entire markets, transactions facilitated by zero or more intermediaries, and goods that can be exchanged individually or in bundles. 

These settings also encompass many common commercial mechanisms including supply chains, distributed markets, security exchanges, and business to consumer auctions. Historically each of these settings has been considered unique and each presented the complex research problem of finding a suitable mechanism (see section \ref{sec:related}  for details).

Second, in addition to the direct-revelation multi-sided auction, our result presents a multi-sided ascending-prices auction that implements the same outcome.
In the theory of one-sided auctions, it is well-known that a second-price direct-revelation auction and an ascending-prices auction are strategically equivalent.  In both auctions, the agent with the highest value wins and pays the second-highest value. However, an ascending-prices auction has the advantage that it is \emph{obviously truthful} (see \citet{li2017obviously} for formal definitions and proofs).
The practical advantage of an obviously-truthful auction is that it is easier for people 
to understand that playing truthfully is best for them, even if they are not experts in game theory. This is particularly important when one deals with complex multi-sided markets with many entities.

\subsection{Paper layout}
Section \ref{sec:model} presents the formal definitions.
Section \ref{sec:1recipe-ones}
presents a special case of our extended multilateral auctions in which each trade requires exactly one agent of each category.
Section \ref{sec:1recipe-general}
presents a more general case of our extended multilateral auctions in which each trade requires a fixed number of agents of each category, but this fixed number may be larger than 1.
Section \ref{sec:related} compares our work to related work.
\ifdefined\ARXIV
Section \ref{sec:experiments} presents some simulation experiments evaluating the performance of our auctions.
\else
The full version presents some simulation experiments evaluating the performance of our auctions.
\fi
Section \ref{sec:future} concludes with some future work directions.
An open-source implementation of our auctions, including example runs and experiments, is available at \url{https://github.com/erelsgl/auctions}.

\section{Preliminaries}
\label{sec:model}
\subsection{Agents and categories}
A \emph{market} is defined by a set of \emph{agents} grouped into different \emph{categories}. 
$N$ is the set of agents, 
$G$ is the set of agent categories, and $N_g$ is the set of agents in category $g\in G$.
The categories are pairwise-disjoint, so
$N = \sqcup_{g\in G} N_g$.

Each deal in the market requires a certain combination of agents. We call a subset of agents that can accomplish a single deal a \emph{procurement-set} (PS).
The \emph{PS recipe} of the market is a vector of size $|G|$, denoted by $\recipe := (r_g)_{g\in G}$, where $r_g\in \mathbb{Z}_+$ for all $g\in G$.
It describes the number of agents of each category that should be in each PS:
each PS should contain $r_1$ agents of category 1, $r_2$ agents of category 2, and so on.
As an example, the PS recipe of a standard two-sided market is $(1,1)$, since there are two agent categories --- buyers and sellers --- and each PS should contain one buyer and one seller.

As another example, consider a market with three categories of agents --- buyers, sellers and transporters, and PS recipe $(1,2,2)$. In such a market, each deal requires a buyer, two sellers and two transporters.

In general, one could think of markets with multiple PS recipes; however, in the present paper we restrict our attention to markets with a single PS recipe, denoted by $\recipe$.

Each agent $i\in N$ has a \emph{value} $v_i\in \mathbb{R}$, which represents the monetary gain of an agent from participating in the trade. The value of an agent is the agent's private information. It may be positive or negative. For example, in a two-sided market, the value of a buyer is typically positive while the value of a seller is typically negative.
The agents are \emph{quasi-linear in money}: the utility of agent $i$ participating in some PS and paying $p_i$ is $u_i := v_i - p_i$.

\subsection{Trades and Gains}
The \emph{gain-from-trade}  of a procurement-set $S$, denoted $GFT(S)$, is the sum of values of all agents in $S$: 
\begin{align*}
\gft(S) := \sum_{i\in S} v_i.
\end{align*}
In a standard two-sided market, the GFT of a PS with a buyer $b$  and a seller $s$ is $v_b - v_s$, since the seller's value is $-v_s$.

Given a market $(N,G,\recipe)$, a \emph{trade} is a collection of pairwise-disjoint procurement-sets. 
I.e, it is a collection of agent subsets, $S_1,\ldots,S_k \subseteq N$, such that for each $j\in[k]$, the composition of agents in $S_j$ corresponds to the recipe $\recipe$. The total GFT is the sum of the GFT of all procurement-sets participating in the trade:
\begin{align*}
\gft(S_1,\ldots,S_k) := 
\sum_{j=1}^k 
\gft(S_j)
\end{align*}
A trade is called \emph{optimal} if its GFT is maximum over all possible trades.

\subsection{Competition}
Our direct-revelation auctions are based on the concept of \emph{competition} between agents.
Given a trade $(S_1,\ldots,S_k)$, 
let $N_{rm} := N\setminus (S_1\cup\cdots\cup S_k)$ be the subset of agents who do not participate in the trade (the ``remaining market'').

Consider a single PS $S_j$, and an agent $i \in S_j$ who belongs to category $g$, i.e, $i \in N_g$.
Then, a subset of agents $T\subseteq N_{rm}$ is called an \emph{external competition} for $i$ if adding $i$ to $T$ yields a PS consistent with the recipe $\recipe$, 
with a positive GFT:
\begin{align*}
GFT (T \cup \{i\}) \geq 0
\end{align*}

In a simple two-sided market, the external competition of a trading buyer is a non-trading seller whose value is sufficiently high such that, combining the trading buyer with the non-trading seller yields a pair with a GFT above $0$.

In a three-sided market with buyers, sellers and mediators, with $\recipe=(1,1,1)$, an external competition of a trading buyer is a pair of a non-trading seller and a non-trading mediator, such that the GFT of the buyer+seller+mediator is at least $0$.

\section{One Agent Per Category}
\label{sec:1recipe-ones}
This section presents our two auctions for a special case in which the single PS recipe in the market is a vector of ones, so each PS must contain a single agent from each category.

Both auctions are parametrized by an ordering on the categories: each of the $|G|!$ possible orderings yields a different auction. The ordering should be fixed in advance and not depend on the agents' values.

We present the auctions using a running example with three categories in the following order: buyers, sellers and mediators. The recipe is $(1,1,1)$. In each category there are five agents. The agents' values are:
\begin{itemize}
\item Buyers: 17, 14, 13, 9, 6.
\item Sellers: -1, -4, -5, -8, -11.
\item Mediators: -1, -3, -4, -7, -10.
\end{itemize}

\subsection{External-competition auction}
\label{sub:1recipe-ones-auction}
The auction requests the agents to report their values, and then proceeds as follows.

\paragraph{Step 1: Optimal trade calculation.} 
Order the agents in each category by descending order of their value.
Combine the highest-value agents in each category to a PS. 
Combine the next-highest-value agents in each category into a PS. Keep constructing PS as long as the GFT of the constructed PS is positive. The resulting set of PS is the optimal trade. 
We denote by $k$ the number of PS in the optimal trade.

In the running example, $k=3$ and the optimal trade contains the following PS:
$(17, -1, -1)$ with GFT $15$, $(14, -4, -3)$ with GFT $7$, and $(13, -5, -4)$ with GFT $4$.
The remaining market, denoted by $N_{rm}$, contains two buyers $9,6$, two sellers $-8,-11$ and two mediators $-7,-10$.

\paragraph{Step 2.} Order the procurement-sets in the optimal trade by ascending GFT, such that,
$\gft(S_1)\leq \cdots \leq \gft(S_k)$.
In the running example, $S_1$ is the PS $(13, -5, -4)$.

\paragraph{Step 3.}
Consider the agents in $S_1$ in the pre-determined order of categories. 
Initialize $i$ to the first agent in $S_1$ by this ordering. In the running example, it is the buyer $13$.

\paragraph{Step 4.}
Look for an external competition to $i$ with a largest GFT. There are two cases.

\paragraph{Case 4a.}
No external competition for $i$ is found. 
Then, $i$ is removed from the trade (and added to $N_{rm}$), and we go back to step 4 with $i$ being the next agent in $S_1$.

In the running example, we consider first the buyer $13$. The maximum GFT of a PS that contains this buyer and agents from $N_{rm}$ is $-2$, for the PS $(13,-8,-7)$. This GFT is negative so it is not considered an external competition. Hence, the buyer $13$ is removed from trade.

\paragraph{Case 4b.}
An external competition for $i$ is found; denote it by $T_1$.
From now on we call this agent $i$ the \emph{pivot agent} and its category the \emph{pivot category}. Denote the pivot category by  $g_o$.
For each $g\neq g_o$, denote by $v^T_g$ the value of the single agent in $T_1\cap N_g$.
Trade prices are calculated as follows:
\begin{itemize}
\item The price $p_g$ for each agent in category $g\neq g_o$ is set to the value $v^T_g$.
\item The price $p_o$ for each agent in category $g_o$ is set to: $p_o := - \sum_{g\neq g_o} v^T_g$.
\end{itemize}

In the running example, the next member of $S_1$ (after the buyer $13$ is removed) is the seller $-5$. 
The maximum GFT of a PS that contains this seller and agents from $N_{rm}$ is $+1$, for the PS $(13,-5,-7)$; note that the removed buyer $13$ participates in this PS.
This GFT is positive so it is an external competition; the pivot category $g_o$ is the sellers' category.

The prices are set to $13$ for the buyers (like the buyer in $T_1$), $-7$ for the mediators (like the mediator in 
$T_1$), and $-6 = -(13-7)$ for the sellers.
The final price-vector is thus $(13, -6, -7)$, i.e., all buyers pay $13$, all sellers receive $6$ and all mediators receive $7$.

\paragraph{Step 5.}
Once the prices are calculated, the final trade is determined as follows:
\begin{itemize}
\item For each category, count the number of members remaining in the trade.
\item In each category with the smallest count, all agents participate in the trade.
\item In each category with a larger count, there is a lottery determining who will participate in the trade.
\end{itemize}

In the running example, there are two  remaining buyers ($17,14$) all of whom trade at price $13$; there are three remaining sellers ($-1,-4,-5$) two of whom (selected at random) trade at price $-6$; similarly, there are three remaining mediators ($-1,-3,-4$) two of whom trade at price $-7$.

Note that selecting a different one of the $6$ category-orders leads to a different outcome.
A-priori, there is no reason to prefer one ordering over the other --- our auction has the same desirable properties (proved below) for any ordering.

The SBBA auction of 
\citet{SegalHalevi2016SBBA}
is a special case of our auction, where the recipe is $(1,1)$.
Their two variants correspond to the two orderings --- buyers-sellers or sellers-buyers.

\subsection{Proof of correctness}
First, note that there must be an agent $i\in S_1$ for whom an external competition exists. In the worst case, when only one last agent of $S_1$ remains in the trade, the other agents of $S_1$ (who were previously removed from trade) form an external competition for this agent. This is because their total GFT is $\gft(S_1)$, which is positive since $S_1$ is in the optimal trade.

\begin{lemma}
\label{lem:ST}
For each category $g\in G$, 
denote by $v^S_g$ the value of the single agent  in $S_1\cap N_g$. Then:
\begin{align*}
\sum_{g \neq g_o} v^T_g \leq 
\sum_{g \neq g_o} v^S_g
\end{align*}
\end{lemma}
\begin{proof}
For $g<g_o$, the agent in $S_1\cap N_g$ had been removed from trade before the pivot was found, and was later used as an external competition for the pivot, so it is the same agent as in $T_1\cap N_g$ and thus $v^T_g = v^S_g$.

For $g>g_o$, the agent in $S_1\cap N_g$ had not been removed from trade, and thus, another (non-trading) agent was used as an external competition. Since the values of non-trading agents are smaller than that of trading agents, $v^T_g \leq  v^S_g$.
\end{proof}

Now we prove the properties of the auction.

\begin{theorem}
\label{thm:1recipe-ones}
The external-competition auction of Subsection \ref{sub:1recipe-ones-auction} is strongly-budget-balanced, individually-rational and dominant-strategy truthful, and its gain-from-trade approaches the optimum when the optimal market size ($k$) approaches $\infty$.
\end{theorem}

\noindent
\emph{Proof.} Strong budget balance is obvious: the price $p_o$ is calculated such that the sum of prices in each PS is $0$.

\paragraph{Individual rationality:} we prove that the price paid by each trading agent is at most the agent's reported value.
\begin{itemize}
\item Each trading agent in a category $g\neq g_o$ pays the value $v^T_g$ of a non-trading agent in the same category $g$. The agents in each category are ordered by descending value, 
and the value of each trading agent is at least as large as the value of each non-trading agent in the same category, so it is at least $v^T_g$.

\item 
Let $v^S_o$ be the value of the pivot agent (who is an agent in $S_1$).
By definition of external competition, the sum of values of agents in $T_1$ plus $v^S_o$ is at least $0$, so 
\begin{align*}
&
v^S_o + \sum_{g \neq g_o}v^T_g \geq 0
\\
\implies &
v^S_o \geq - \sum_{g \neq g_o}v^T_g = p_o
\end{align*}
Since agents are ordered by descending value, the values of other trading agents are at least $v^S_o$ which is at least $p_o$.
\end{itemize}

\paragraph{Truthfulness:}
By Myerson's theorem, it is sufficient to prove that the choice rule is monotone, and each trading agent pays his/her threshold value.

Monotonicity is obvious: an agent increasing his reported value (while other reports are fixed) is more likely to participate in the optimal trade, more likely to have an external competition, and thus more likely to remain in the trade.

To calculate the threshold value of an agent $i$ from category $g$, we consider three cases, depending on the fixed ordering of the categories:
\begin{enumerate}
\item The category $g$ comes before the pivot-category $g_o$. This means that an agent from $g$ had been removed from $S_1$ before the pivot was found.
All agents whose value is higher than $v^T_g$ are in PS $S_j$ for $j\geq 2$, they do not affect the auction in any way, and they remain in the trade. Any such agent whose value drops below $v^T_g$, replaces the $v^T_g$ agent in the PS $S_1$, and has no external competition, and so is removed from the trade. Therefore, $v^T_g$ is a threshold-value for all agents of $g$, and indeed $p_g = v^T_g$.
\item $g = g_o$. 
All agents whose value is higher than $v^S_o$ (the value of the pivot agent) are in PS $S_j$ for $j\geq 2$, they do not affect the auction in any way, and they remain in the trade. 
Consider an agent of $g_o$ whose value $v_o$ drops below $v^S_o$ but above $p_o$ (recall that $v^S_o\geq p_o$).
We claim that this agent remains in the trade.
First, $v_o$ is still in the optimal trade (it
replaces the pivot agent in $S_1$), since:
\begin{align*}
v_o + \sum_{g \neq g_o} v^S_g 
&\geq 
v_o +  \sum_{g \neq g_o} v^T_g \text{~~(by Lemma \ref{lem:ST})}
\\
& = 
v_o - p_o \text{~~(by definition of $p_o$)}
\\
& \geq 0 \text{~~(by assumption on $v_o$)},
\end{align*}
so the GFT of $v_o$ plus the other agents in $S_1$ is still above 0.
Second, $T_1$ is still an external competition for $v_o$, since:
\begin{align*}
v_o + \sum_{g \neq g_o} v^T_g = 
v_o - p_o \geq 0.
\end{align*}
But, once $v_o$ drops below $p_o$, the set $T_1$ is no longer an external competition, so the agent is removed from the trade. 
Hence, $p_o$ is a threshold value for all agents of $g_o$.
\item The category $g$ comes after the pivot-category $g_o$.
This means that no agent from $g$ had been removed from the trade before the pivot was found.
As shown in Lemma \ref{lem:ST}, 
in this case $v^S_g\geq v^T_g$.
All agents whose value is higher than $v^S_g$ are in PS $S_j$ for $j\geq 2$,  they do not affect the auction in any way, and they remain in the trade.

Consider an agent of $g$ whose value $v_g$ drops below $v^S_g$ but above $v^T_g$.
We claim that this agent remains in the optimal trade (it replaces the agent $v^S_g$ in $S_1$), since:
\begin{align*}
GFT(S_1) - v^S_g + v_g
&\geq 
GFT(S_1) - v^S_g + v^T_g 
\\
&=
v^S_o + \sum_{g \neq g_o} v^S_g - v^S_g + v^T_g
\\
&\geq 
v^S_o + \sum_{g \neq g_o} v^T_g \text{~~(by Lemma \ref{lem:ST})}
\\
&\geq  0 \text{~~($T_1$ is external competition)}
\end{align*} 
so the GFT of $v_g$ plus the other agents in $S_1$ is above 0.

But, once $v_g$ drops below $v^T_g$, it is replaced by the agent $v^T_g$ in $S_1$, and does not enter the optimal trade. 
Hence, $v^T_g$ is a threshold-value for all agents of $g$, and $p_g = v^T_g$.
\end{enumerate}

\paragraph{Gain-from-trade:}
For each $g\in G$ and $j\range{1}{k}$, denote by $v^j_{g}$ the value of the single agent in $N_g\cap S_j$. Then the optimal GFT is:
\begin{align*}
OPT = \sum_{g\in G}\sum_{j=1}^k v^j_{g}
\end{align*}
If no traders are removed, then all these $k$ PS are trading, and the GFT equals OPT.
If some traders are removed, they are removed from $S_1$ which is the least profitable PS.
In this case, $k-1$ deals are made, where in each deal, the trader from each category $g$ is:
\begin{itemize}
\item If $g$ is before the pivot --- one of the $k-1$ high-value traders in $g$;
\item If $g$ is the pivot or after the pivot --- one of the $k$ high-value traders in $g$, selected at random.
\end{itemize}
Hence, the expected GFT is at least:
\begin{align*}
&\sum_{g < g_o}\sum_{j=2}^k v^j_{g}
+
\sum_{g \geq g_o} \frac{k-1}{k} \sum_{j=1}^k v^j_{g}
\\
\geq &
\sum_{g \in G} \frac{k-1}{k} \sum_{j=1}^k v^j_{g}
\\
= & \frac{k-1}{k} OPT.
\end{align*}

\subsection{Ascending-prices auction}
\label{sub:1recipe-ones-clock}
Our ascending-prices auction holds a price $p_g$ for each category $g\in G$.
All prices are initialized to $-\infty$, and initially all agents are in the trade (since every agent will be happy to pay $-\infty$).
While the prices increase, each agent in category $g$ with value $v_g$ remains in the trade as long as $p_g < v_g$, and exits the trade when $p_g > v_g$ (since the prices increase monotonically, agents never return to the trade after exiting).
When $p_g = v_g$, the agent is indifferent between trading and not trading; for simplicity, we assume that in this case the agent does not trade.
Also, for simplicity we assume that the agents' valuations are \emph{generic} in the sense that, for each category $g\in G$, all agents have different values. 
During the presentation of the ascending auction, we use the same running example as in Subsection \ref{sub:1recipe-ones-auction}.

\paragraph{Step 1: Initialization.}
For each category $g$, count the number of agents in $N_g$; let $n_{\min}$ be the size of the smallest category. For each category $g$ with more than $n_{\min}$ agents, increase the price $p_g$ such that some agents leave the trade, until the number of remaining agents in all categories is $n_{\min}$.

In the running example, this step is not needed since initially there are 5 agents in each category.

\paragraph{Step 2.}
Loop over the categories in the pre-specified order. For each category $g$, increase $p_g$ continuously until one of the following happens:

(a) an agent from category $g$ exits the trade, or ---

(b) the sum of prices increases to zero: $\sum_{g\in G}p_g = 0$.

In case (a), repeat the step with the next category (after the last category, return to the first one).
If a category becomes empty, the auction stops and there is no trade.

In case (b), stop and have the agents trade in the final prices: each agent in category $g\in G$ trades at price $p_g$. If, in some category, there are more remaining agents than in other categories, then a lottery is used to select who will trade.

In the running example, at the first round, the buyers' price increases to 6, the sellers' price increases to -11, the mediators' price increases to -10; after the first round, there are 4 agents remaining in each category, and the sum of prices is still negative, so we continue.
At the second round, the prices increase to 9, -8, -7 and the sum is still negative.
At the third round, the buyers' price increases to 13
and the sellers' price is increased towards -5, but when it hits -6, the sum of prices becomes 0 so the auction stops.
The final trade is exactly the same as in the external-competition auction.

\begin{theorem}
\label{thm:1recipe-ones-clock}
The ascending-prices auction of Subsection \ref{sub:1recipe-ones-clock} is strongly-budget-balanced, individually-rational and
obviously truthful, and its gain-from-trade approaches the optimum when the optimal market size ($k$) approaches $\infty$.
\end{theorem}
\noindent
\begin{proof}
SBB and IR are immediate from the description.


As for obvious-truthfulness:
\citet{li2017obviously}
defines a strategy $S$ as obviously-dominant (for a given agent) if ``for any deviating strategy $T$, starting from an earliest information set where $S$ and $T$ diverge, the best possible outcome from $T$ is no better than the worst possible outcome from $S$''. 
We show that, in the ascending-prices auction, for each agent $i$ in category $g$, the strategy $S$ of exiting when $p_g = v_i$ is obviously-dominant.

The worst outcome from $S$ has a value of $0$.
We now show that, for any deviation $T$, the best possible outcome from $T$ when $S$ and $T$ diverge has a value of at most $0$.
Indeed, if $T$ is exiting too early (at some $v_i' < v_i$), then the point at which $S$ and $T$ diverge is when $p_g  = v_i'$, and at that point the outcome from $T$ has a value of $0$.
If $T$ is exiting too late (at some $v_i'' > v_i$), then the point at which $S$ and $T$ diverge is when $p_g  = v_i$, 
and at that point all possible outcomes from $T$ have a value of $0$ or less.

We now analyze the gain-from-trade.

Let $g_o$ be the category in which the protocol stops. 
Let $n_o$ be the number of traders of this category that remain in the trade.
Then, in each category $g<g_o$, there are $n_o-1$ remaining traders, and in each category $g\geq g_o$, there are $n_o$ remaining traders.

Recall that $k$ is the number of deals in the optimal trade; we claim that $k = n_o$:
\begin{itemize}
\item First, suppose that the price $p_g$ of each category $g\geq g_o$ is increased up to the value of the next agent in $g$ (who did not exit the trade in the actual auction). Since the auction stopped when the sum of prices hit 0, the sum of prices after the increase is positive.
Each price $p_g$ equals the $n_o$-th highest value in category $g$. This means that there are at least $n_o$ procurement-sets with a positive GFT, so $k\geq n_o$.
\item Second, suppose that the price $p_g$ of each category $g\leq g_o$ is decreased down to the value of the previous agent in $g$ (who did exit the trade in the actual auction). Now the sum of prices is negative. 
Each price $p_g$ equals the $(n_o+1)$-th highest value in category $g$. 
This means that there are not $(n_o+1)$ procurement-sets with a positive GFT, so $k < n_o+1$. Hence, $k = n_o$.
\end{itemize}
So at least $k-1$ deals are done. From here, the proof is identical to the gain-from-trade proof in Theorem \ref{thm:1recipe-ones}. 
\end{proof}

\section{General Procurement-Set Recipes}
\label{sec:1recipe-general}
This section extends the previous one by allowing the PS recipe to be an arbitrary vector of positive integers, rather than just a vector of ones. 
For each category $g$ there is an integer $r_g\geq 1$, and every PS must contain exactly $r_g$ traders from this category.

\ifdefined\EXTERNALCOMPETITION
\else
The external-competition auction can be extended to the setting of an arbitrary vector of positive integers. However, as the proof is somewhat involved, we choose to focus here on the ascending-prices auction extension.
\fi

We present the mechanisms using a running example in which there are two categories --- buyers and sellers, and the recipe is $(1,2)$, so that each PS should contain one buyer and two sellers. The market contains:
\begin{itemize}
\item Five buyers with values: 17, 14, 13, 9, 6.
\item Nine sellers with values: -1, -2, -3, -4, -5, -7, -8, -10, -11.
\end{itemize}

\ifdefined\EXTERNALCOMPETITION
For the external-competition auction we present two larger running examples:
a second example with three categories, recipe $(2,2,3)$, and valuations:
\begin{itemize}
\item Buyers: $[17, 16, 15, 14, 13, 12, 10, 6]$
\item Mediators: $[-3, -4, -5, -6, -7, -8, -9, -10]$
\item Sellers: $[-1, -2, -3, -4, -5, -6, -7, -8]$
\end{itemize}
and a third example (for illustrating the difference in Step 1 and 2 from Section \ref{sec:1recipe-ones}) with two categories, recipe $(3,2)$, and valuations:
\begin{itemize}
\item Buyers: $[20, 18, 16, 9, 2, 1]$
\item Sellers: $[-2, -4, -6, -8, -10, -12, -14]$
\end{itemize}

\else
Note that the optimal trade in this setting can be calculated just like in Section \ref{sec:1recipe-ones}: the agents in each category are ordered by descending value, and then grouped greedily into procurement-sets.
In the running example, the optimal trade contains three PS:
$(17; -1, -2)$ with GFT $14$,
$(14; -3, -4)$ with GFT $7$,
and $(13; -5, -7)$ with GFT $1$.
\fi

\ifdefined\EXTERNALCOMPETITION
\subsection{External-competition auction}
\label{sub:ec-auction}

\paragraph{Step 1: PS Construction.} 
Order the agents in each category by descending order of their value.
Combine the highest-value agents in each category to a PS. 
Combine the next-highest-value agents in each category into a PS. Keep constructing PS as long as there are enough agents to follow the recipe of a PS. Unlike Section \ref{sec:1recipe-ones}, the GFT of a constructed PS can be negative. The set of PS with positive GFT is the optimal trade.  
Like Section \ref{sec:1recipe-ones} we denote by $k$ the number of PS in the optimal trade. We denote by $w$ the number of PS we constructed in total (including the negative GFT ones if exist).

In the first running example, $k=3$ and $w=4$. The optimal trade contains three PS:
$(17; -1, -2)$ with GFT $14$,
$(14; -3, -4)$ with GFT $7$,
and $(13; -5, -7)$ with GFT $1$.
The entire PS set contains, besides the three optimal PS, also 
$(9; -8, -10)$ with GFT $-9$.
The remaining market, that we denote by $W_{rm}$, contains the buyer $6$ and the seller $-11$ (the agents not in any PS, whether negative or positive).

In the second running example, $k=w=2$. The optimal trade contains two PS:
$(17, 16; -3, -4; -1, -2, -3)$ with GFT $20$,
and $(15, 14; -5, -6; -4, -5, -6)$ with GFT $3$.
The remaining market $W_{rm}$ contains the buyers $13, 12, 10, 6$, the mediators $-7, -8, -9, -10$ and the sellers $-7, -8$.

In the third running example, $k=1$ and $w=2$. The optimal trade contains the following PS:
$(20, 18, 16; -2, -4)$ with GFT $48$. 
The PS set contains also the following PS: 
and $(9, 2, 1; -6, -8)$ with GFT $-2$.
The remaining market $W_{rm}$ contains no buyers, and the sellers $-10,-12, -14$.

\paragraph{Step 2.} Order all the procurement-sets (including the negative ones) by ascending GFT, such that,
$\gft(S_1)\leq \cdots \leq \gft(S_w)$.
In the first running example, $S_1$ is the PS $(9; -8, -10)$.
In the second running example, $S_1$ is the PS $(15, 14; -5, -6; -4, -5, -6)$.
In the third running example, $S_1$ is the PS $(9, 2, 1; -6, -8)$.

\paragraph{Step 3.}
Consider the agents in $S_1$ in the pre-determined order of categories. 
Inside each category, order the agents by increasing value. In the first running example, 
the order is
buyer $9$, seller $-10$, seller $-8$.
In the second running example, the order is buyer $14$, buyer $15$, mediator $-6$, mediator $-5$, seller $-6$, seller $-5$, seller $-4$.
In the third running example, the order is buyer $1$, buyer $2$, buyer $9$, seller $-8$, seller $-6$.

Let $i$ be the first agent in $S_1$ by this ordering, which has the lowest value of his category in $S_1$. In the first running example $i$ is buyer $9$, in the second running example $i$ is buyer $14$ and the third running example $i$ is buyer $1$.

\paragraph{Step 4.}
Look for an external competition for $i$ with a largest GFT. 
Importantly, an external competition is constructed using a \emph{single} agent from each category $g$ (including $i$'s category), which is duplicated $r_g$ times.
The agents from categories different than $i$ are selected  from the remaining market $W_{rm}$.

For example, to construct an external competition for buyer $9$ in the first running example, we take the highest-valued non-trading seller $-11$, create two copies of it, and get the PS $(9,-11,-11)$. 
Its GFT is negative so it cannot serve as external competition.
In the second running example, to construct an external competition for buyer $14$, we take the highest-valued non trading mediator $-7$, create two copies of it, then take the highest-valued non trading seller $-7$, create three copies of it and get the PS $(14,14;-7,-7;-7,-7,-7)$. Its GFT is negative so it cannot serve as external competition. 
In the third running example, to construct an external competition for buyer $1$, we take the highest-valued non trading seller $-10$, create two copies of it and get the PS $(1,1,1;-10,-10)$. Its GFT is negative so it cannot serve as external competition. 

Now, there are two cases.

\paragraph{Case 4a.}
No external competition for $i$ is found. 
Then, $i$ is removed from the trade (and added to $W_{rm}$), and we go back to step 4 with the next agent in $S_1$.

In the first running example, as explained above, there is no external competition for buyer $9$ so it is removed from trade.
Next, we consider the seller $-10$.
We create two copies of it and add the highest-valued non-trading buyer ($9$). We get the PS $(9,-10,-10)$. 
Its GFT is negative so it cannot be an external competition. Hence, seller $-10$ is removed from trade.
Similarly, we remove the seller  $-8$, 
the buyer $13$, 
and the seller $-7$, 

In the second running example, as explained above, there is no external competition for buyer $14$ so it is removed from trade.
Similarly we remove the buyer $15$.
and the mediators $-6$ and $-5$.

In the third running example, as explained above, there is no external competition for buyer $1$ so it is removed from trade. Similarly we remove the buyer $2$.

\paragraph{Case 4b.}
An external competition for $i$ is found; denote it by $I_{co}$.
We denote the pivot category by $g_o$.

In the first running example, the next agent to check is seller $-5$. The buyer $13$ serves as an external competition for it, since the GFT of $(13, -5, -5)$ is positive. Hence the pivot category is the category of sellers.

In the second running example, the next agent to check is seller $-6$. The buyer $15$ and seller $-5$ and their copies serve as an external competition for it, since the GFT of $(15, 15, -5, -5, -6, -6, -6)$ is positive. Hence the pivot category is the category of sellers.

In the third running example, the next agent to check is buyer $9$. The seller $-10$ serves as an external competition for it, since the GFT of $(9, 9, 9; -10, -10)$ is positive. Hence the pivot category is the category of buyers.

The trade-prices are calculated as follows:
\begin{itemize}
\item The price $p_g$ for each agent in category $g\neq g_o$ is set to the value $v_g$ of the single agent in $I_{co}\cap N_g$.
Note that, even if $r_g > 1$, there is still a single agent in $I_{co}$, which we duplicate $r_g$ times.

In the first running example, the price for a buyer is set to $13$.
In the second running example, the price for a buyer is set to $15$ and the price for a mediator is set to $-5$.
In the third running example, the price for a seller is set to $-10$.

\item The price $p_o$ for each agent in category $g_o$ is set to: $p_o := - (\sum_{g\neq g_o} r_g \cdot v_g) / r_o$. 

In the first running example, the price for each seller is set to $- 13 / 2 = - 6.5$.
In the second running example, the price for each seller is set to $-(15*2+(-5)*2)/3=-6.67$.
In the third running example, the price for each buyer is set to $-((-10)*2)/3=6.67$. 
\end{itemize}

The trade is calculated just like in Section \ref{sec:1recipe-ones}. 
In the first running example, there are two remaining buyers ($17,14$) and all of them trade at price $13$; there are five remaining sellers ($-1,-2,-3,-4,-5$) and four of them, selected at random, trade at price $-6.5$.
In the second running example, there are two remaining buyers ($17,16$) and all of them trade at price $15$; there are two remaining mediators ($-3,-4$) and all of them trade at price $-5$; there are six remaining sellers ($-1,-2,-3,-4,-5,-6$) and three of them, selected at random, trade at price $-6.67$.
In the third running example, there are four remaining buyers ($20, 18, 16, 9$) and three of them, selected at random, trade at price $6.67$; There are four remaining sellers ($-2, -4, -6, -8$) and two of them, selected at random, trade at price $-10$.

As an additional illustration of the mechanism, suppose it is run on the same market as above, but with a different category order: sellers before buyers. 
Then, in the first running example, the agents in $S_1$ are checked in order $(-10, -8, 9)$. 
Seller $-10$ is considered first. The highest-valued non-trading buyer is $6$, but the GFT of the PS $(-10,-10,6)$ is negative, so $-10$ is removed. Similarly, $-8$ is removed and so is buyer $9$, seller $-7$ and seller $-5$.
Now, buyer $13$ has an external competition --- the removed seller $-5$, since the GFT of $(13,-5,-5)$ is positive. 
The seller price is set to $-5$ and the buyer price is set to $-(-5\cdot 2) = 10$. 
There are four remaining sellers $(-1,-2,-3,-4)$ and all of them trade at price $-5$; there are three remaining buyers $(17,14,13)$ and two of them, selected at random, trade at price $+10$.

In the second running example if we consider a different category order, for instance: mediators then sellers and last buyers.
Then, the agents in $S_1$ are checked in order $(-6, -5, -6, -5, -4, 14, 15)$. 
Mediator $-6$ is considered first. The highest-valued non-trading buyer is $13$ and the highest-valued non-trading seller is $-7$, but the GFT of the PS $(-6, -6, -7, -7, -7, 13, 13)$ is negative, so mediator $-6$ is removed. Similarly, mediator $-5$ is removed. Then seller $-6$ is considered. The highest-valued non-trading buyer is $13$ and the highest-valued non-trading mediator is $-5$, but the GFT of the PS $(-5, -5, -6, -6, -6, 13, 13)$ is negative, so seller $-6$ is removed. 
Now, seller $-5$ has an external competition --- the removed mediator $-5$ and non-trading buyer $13$, since the GFT of $(-5, -5, -5, -5, -5, 13, 13)$ is positive. 
The buyer price is set to $13$, the mediator price is set to $-5$ and the seller price is set to $-(13\cdot 2+(-5)\cdot 2)/3 = -5.3$. 
There are four remaining buyers $(17, 16, 15, 14)$ and two of them, selected at random, trade at price $13$; there are two remaining mediators  $(-3, -4)$ and all of them trade at price $-5$; there are five remaining sellers $(-1, -2, -3, -4, -5)$ and three of them, selected at random, trade at price $-5.3$.

In the third running example, the agents in $S_1$ are checked in order $(-8, -6, 1, 2, 9)$. 
Seller $-8$ is considered first. There are no buyers that are not part of a PS, so $-8$ is removed. Similarly seller $-6$ is removed.   Buyer $1$ is considered next. The highest-valued non-trading seller is $-6$, but the GFT of the PS $(1, 1, 1, -6, -6)$ is negative, so buyer $1$ is removed. Similarly buyer $2$ is removed.
Now, buyer $9$ has an external competition --- the removed seller $-6$, since the GFT of $(9, 9, 9, -6,-6)$ is positive. 
The seller price is set to $-6$ and the buyer price is set to $-(-6\cdot 2)/3 = 4$. 
There are two remaining sellers $(-2,-4)$ and all of them trade at price $-6$; there are four remaining buyers $(20, 18,16, 9)$ and three of them, selected at random, trade at price $+4$.

\subsection{Proof of correctness}
\label{sub:ec-proof}

First we prove a claim regarding the reduction processes that will help us to prove the economic properties. Below, by \emph{negative PS} we mean a procurement-set with a negative GFT.
\begin{claim}
\label{claim-S1}
The trade-reduction process either removes all agents in negative PS, or leaves a strict subset of a single negative PS, which is the negative PS with the highest GFT.
\end{claim}

\begin{proof}
Assume to the contrary that $S_j$ is a negative PS and none of its agents were removed from the trade. 
Consider the pivot agent --- the agent with lowest value in $S_j\cap N_{g_o}$.
Since no agent of $S_j$ were removed, this agent has an external competition.
By definition of external competition,  it has a positive GFT,
but it is made of agents that are not in the trade. Since agents are considered by ascending value, the values of the agents in the external competition are lower than the values of agents from the same categories in $S_j$.
Hence, the GFT of $S_j$ is positive too --- contradicting the assumption that $S_j$ has a negative GFT.

Since the agents are removed by the order of the PS, which is an ascending order of GFT, if a PS (or a part of it) is removed, then all PS with a lower GFT are removed as well.
\end{proof}

We conclude from claim \ref{claim-S1} that the pivot agent will be found either in $S_{w-k+1}$ --- the PS with the smallest positive GFT, or in $S_{w-k}$ --- the PS with the largest negative GFT
(note that in Section \ref{sec:1recipe-ones} 
the pivot agent was always found in the PS with the smallest positive GFT).
 
\paragraph{Individual rationality:}
we prove that the price of each trading agent is at most the agent's reported value.
\begin{itemize}
\item Each trading agent in a category $g\neq g_o$ pays the value $v_g$ of a non-trading agent in the same category $g$. The agents in each category are ordered by value, 
and the value of each trading agent is at least as large as the value of each non-trading agent in the same category, so it is at least $v_g$.

\item Each trading agent in the pivot category pays $p_o$, which is calculated such that:
$ r_o \cdot p_o+ \sum_{g\neq g_o} r_g \cdot v_g = 0$.
Let $v_o$ be the value of the pivot agent.
By definition of external competition, 
$r_o \cdot v_o + \sum_{g\neq g_o} r_g \cdot v_g \geq 0$, 
so $v_o\geq p_o$.
Since agents are ordered by value, the values of other trading agents are at least $v_o$ which is at least $p_o$. 
\end{itemize}

\paragraph{Truthfulness.}
As in Section \ref{sec:1recipe-ones}, 
the monotonicity of the choice-rule is obvious, and we focus on proving that each trading agent pays his/her threshold value.

Assume first that the agent with $v_o$ is found in PS $S_{w-k}$ (the proof for the case $S_{w-k+1}$ is similar).
To calculate the threshold value of an agent $i$ from category $g$, we consider three cases w.r.t. the fixed category ordering:
\begin{enumerate}
\item The category $g$ comes before the pivot-category $g_o$. This means that $r_g$ agents from $g$ had been removed from $S_{w-k}$ before the pivot was found.
The highest-valued removed agent, whose value was denoted by $v_g$, participated in the external competition.
All agents whose value is higher than $v_g$ are in PS $S_j$ for $j\geq (w-k)+1$, they do not affect the mechanism in any way, and they remain in the trade. Any such agent whose value drops below $v_g$, replaces the $v_g$ agent in the PS $S_{w-k}$, has no external competition, and is removed from the trade. Therefore, $v_g$ is a threshold-value for all trading agents of $g$, and indeed $p_g = v_g$.

\item 
$g = g_o$.
All agents whose value is higher than $v_o$ (the value of the pivot agent) are in PS $S_j$ for $j\geq w-k$, they do not affect the mechanism in any way, and they remain in the trade.
Consider such agent $i$ whose value $v_i$ drops below
$v_o$ but above $p_o$ (recall that $v_o\geq p_o$).
We claim that $I_{co}$ is still an external competition for this agent.
This is because, when each agent in $I_{co}$ is duplicated in accordance with the PS recipe, the sum of their values is $-r_o\cdot p_o$ (by definition of $p_o$). 
Adding $r_o$ copies of $v_i$ gives a total GFT of $r_o\cdot v_i - r_o\cdot p_o = r_o\cdot (v_i-p_o)\geq 0$.
Hence, this agent remains in the trade.
But an agent whose value drops below $p_o$
no longer has an external competition, since now $v_i-p_o < 0$, and thus is removed from the trade. 
Therefore, $p_o$ is indeed a threshold-value for all agents of $g_o$.

\item The category $g$ comes after the pivot-category $g_o$.

This means that no agent from $N_g \cap S_j$, for any $j\geq w-k$, had been removed from the trade before the pivot was found.
So $v_g$ must be the highest value of a member of $S_{w-k-1}\cap N_g$.
Let $v'_g$ be the next-highest value of an agent in $g$ --- the lowest value of a member in $S_{w-k}\cap N_g$. 
All agents whose value is higher than $v_g'$ are in PS $S_j$ for $j\geq w-k$,  they do not affect the mechanism in any way, and they remain in the trade.

Any such agent $i$ whose value $v_i$ drops below
$v_g'$ but above $v_g$
becomes a member of $S_{w-k}$, since the highest-valued $r_g \cdot (w-k)$ agents of $N_g$ are divided into $w-k$ procurement-sets, and $v_i$ is still among  the highest-valued $r_g \cdot (w-k)$ agents (albeit the lowest-valued of them).
Since no agent of $S_{w-k}\cap N_g$ is removed in this case, $v_i$ remains in the trade.

In contrast,  if $v_i$ drops below $v_g$, then
it becomes a member of $S_{w-k-1}$.
Since $v_g$ did not have an external competition, $v_i$ too does not have an external competition and is removed from trade.

Therefore, $v_g$ is a threshold-value for all agents of $g$, and indeed $p_g = v_g$.
\end{enumerate}

\paragraph{Gain-from-trade.}
Consider first the case in which all negative PS have been removed. Then the situation is similar to the one in 
Theorem \ref{thm:1recipe-ones}: all agents in the $k-1$ highest positive PS participate in the lottery, plus some agents in the $k$th-highest positive PS (denoted here by $S_{w-k+1}$).
Therefore, for each category $g$ and for each $j\geq w-k+2$,
each member of $S_j\cap N_g$ participates in the lottery.
In the lottery, at least $(k-1)\cdot r_g$ agents of each category $g$ are selected at random out of at most $k\cdot r_g$. 
Therefore, each agent has a chance of at least $((k-1)\cdot r_g)/(k\cdot r_g)$ to participate in the trade.
The expected contribution to the GFT of each such agent with value $v_i$ is 
$v_i(k-1)/k$.
By the linearity of expectation, we can sum over all agents in $N_g$
and all $k-1$ PS. We get that the expected GFT is at least $(k-1)/k$ of OPT.

Consider now the case in which some agents in a negative PS are not removed. By Claim \ref{claim-S1}, all these agents must be in the PS $S_{w-k}$, which is the $(k+1)$-th highest PS.
During the lottery, $k\cdot r_g$ agents of each category $g$ are selected at random, and $k$ new PS are formed. 
We can assume w.l.o.g. that all agents from 
$S_{w-k}$ (if any) are put in a single new PS.
Note that this new PS must have a positive GFT, 
since 
all non-removed agents (even the agents in $S_{w-k}$) have an external competition. This means that they have a positive GFT when combined with some non-trading agents, whose values are even lower than the values of trading agents.
The remaining $k-1$ new PS contain only agents from the $k$ positive PS.
Each agent from these $k$ PS has a chance of at least $(k-1)/k$ to participate in the trade. Hence the expected GFT is at least $(k-1)/k$ of OPT, as above.

\subsection{Ascending-prices auction}
\label{sub:general-clock}
Similarly to subsection  
\ref{sub:1recipe-ones-clock},
the auction maintains a price $p_g$ for each category $g$.
All prices are initialized to $-\infty$, and initially all agents are in the trade.

\paragraph{Step 1: Initialization.}
For each category $g$, 
let $c_g := \floor{|N_g| / r_g}$; this is the largest number of PS that can be composed of agents of category $g$.
Let $c_{\min} := \min_{g\in G} c_g$; this is the largest number of PS that can be composed of the existing agents. 
For each category $g$ for which 
$\floor{|N_g| / r_g} > c_{\min}$,
increase the price $p_g$ such that some agents leave the trade, until the number of remaining agents in each category $g$ decreases such that $\floor{|N_g| / r_g} = c_{\min}$.
Note that, when $g$ is the category for which the minimum of $c_g$ is attained, 
we already have $\floor{|N_g| / r_g} = c_{\min}$ without any increase, so the price of this category (at least) remains $-\infty$, and the initial price-sum is negative.

Initialize $c := c_{\min}$. Informally, $c$ is the number of procurement-sets that we aim to construct from the traders currently in the market.
Note that initially, for every category $g$, we have $\floor{|N_g| / r_g} = c_{\min}$ so $|N_g| \geq r_g \cdot c$.
 
In the running example, $c_{buyers} = \floor{5/1} = 5$ and $c_{sellers} = \floor{9/2} = 4$ and $c_{\min}=4$.
In the initialization step, $p_{buyers}$ increases to $6$ so that the low-value buyer leaves. The market now has $4$ buyers and $9$ sellers. The value of $c$ is initially $4$.

\paragraph{Step 2.}
Loop over the categories in the pre-specified order. For each category $g$, increase $p_g$ until one of the following happens:

(a) The number of agents in $g$ drops to $r_g \cdot c$, or ---

(b) The \emph{weighted} sum of prices increases to zero, where the weights are the recipe constants, i.e.: $\sum_{g\in G}r_g\cdot p_g = 0$.

In case (a), repeat the step with the next category in the pre-specified order. After the last category, 
set $c := c - 1$ and cycle back to the first category.
If $c$ drops to $0$ then the auction ends and there is no trade.

In case (b), the auction terminates and the agents trade in the final prices. If, in some category, there are more remaining agents than needed to fill the procurement-sets, then a lottery is used to select who will trade.

In the running example,
at the first round, $p_{sellers}$ increases to $-11$ so that one seller leaves. Now there are $4$ buyers and $8$ sellers so exactly $4$ PS are supported. However, the price-sum is $6 + 2\cdot (-11) < 0$.
Hence, we decrease $c$ to $3$ and continue.

In the second round, $p_{buyers}$ increases to $9$ such that one buyer leaves, 
and $p_{sellers}$ increases to $-8$ such that two sellers leave. The weighted sum of prices is $9 + 2\cdot(-8) < 0$, so we decrease $c$ to $2$ and continue.

In the third round, $p_{buyers}$ increases to $13$ such that one buyer leaves. 
We start increasing $p_{sellers}$ towards $-5$ such that two sellers would leave, but during the increase, the price hits $-6.5$, and then the sum of prices is $13 + 2\cdot (-6.5) = 0$, so the auction stops.

There are now two remaining buyers ($17,14$) and five remaining sellers ($-1,-2,-3,-4,-5$). 
All remaining buyers trade and pay $13$; 
$4$ out of $5$  remaining sellers are selected at random,  they trade and receive $6.5$.

\begin{theorem}
The ascending-prices auction of Subsection \ref{sub:general-clock} is strongly-budget-balanced, individually-rational and
obviously truthful, and its gain-from-trade approaches the optimum when the optimal market size ($k$) approaches $\infty$.
\end{theorem}
\noindent
\begin{proof}
SBB, IR and truthfulness are obvious --- just as in Theorem \ref{thm:1recipe-ones-clock}.
We now analyze the gain-from-trade.
Let $c_*$ be the final value of $c$.
Note that the number of agents remaining in each category $g$ is at least $r_g\cdot c_*$ and at most $r_g\cdot (c_*+1)$.

We claim that $c_*\leq k \leq c_*+1$:
\begin{itemize}
\item First, suppose that the price $p_g$ of each category $g$ is increased such that exactly $r_g\cdot c_*$ agents remain. Now the sum of prices is positive, and each price $p_g$ equals the $(r_g\cdot c_* + 1)$-th highest value in category $g$. This means that there are at least $c_*$ procurement-sets with a positive GFT, so $k\geq c_*$.
\item Second, suppose that the price $p_g$ of each category $g$ is decreased to its value at the end of round $c_*+1$, such that 
exactly $r_g\cdot (c_*+1)$ agents remain. 
Now the sum of prices is negative. 
Each price $p_g$ equals the $(r_g\cdot c_* + r_g + 1)$-th highest value in category $g$.
The expression $(r_g\cdot c_* + r_g + 1)$
is at most $r_g\cdot (c_* + 2)$, with equality holding iff $r_g=1$.
 This means that there are not $(c_*+2)$ procurement-sets with a positive GFT, so $k \leq c_*+1$. 
\end{itemize}
Hence, in each category $g$, the number of agents participating in the lottery is at least $(k-1)\cdot r_g$ and at most $(k+1)\cdot r_g$.
In the lottery, at least $(k-1)\cdot r_g$ 
agents from $N_g$ are selected.
Therefore, each agent from the $k\cdot r_g$ agents in the optimal trade has a chance of at least $(k-1)/(k+1)$ to win the lottery. Therefore, the expected GFT is at least 
$(k-1)/(k+1)$ of OPT.
\end{proof}

\section{Related Work}
\label{sec:related}
The literature on two-sided markets is large and we do not attempt to cover it here; 
see e.g.
\cite{Brustle2017Approximating,babaioff2018best,babaioff2020bulow} for some recent references.
Below we focus on auctions
for markets with three or more sides.

\citet{babaioff2004concurrent} handle a multi-sided market using multiple two-sided sub-markets, where each sub-market hosts an independent double-auction.
Their example is the \emph{lemonade industry}, which consists of lemon pickers, squeezers, and drinkers.
In our auction, all three categories bid together in a single centralized auction with recipe $(1,1,1)$; 
in their setting, there are three different double-auctions, one each for pickers, squeezers and drinkers. The sub-markets are constructed so that the optimal number of deals is the same in all of them. So if the double-auction mechanisms make deals as a function of the optimal number only, their protocol is guaranteed to have a material balance. However, their protocol does not preserve SBB:
in \url{https://github.com/erelsgl/auctions} we have a runnable example in which the same SBB double-auction is used in all sub-markets, but the combined outcome is not SBB.
\citet{babaioff2004concurrent} do present a SBB variant of their mechanism, but only in expectation, and it requires a prior distribution on the agents' valuations.  In contrast, our mechanism is SBB with probability 1, and requires no prior.

\citet{babaioff2005incentivecompatible} extend the above work from a linear supply-chain to an arbitrary directed acyclic graph. For example, they consider a market in which pickers sell lemons to squeezers, sugar-makers and squeezers sell to manufacturers, and manufacturers sell lemonade to drinkers. It still does not guarantee SBB.

\citet{chen2005efficient} consider a supply-chain auction
with a sole buyer and single item-kind, but there are different producers in different supply-locations. The buyer needs a different quantity of the item in different demand-locations. The buyer conducts a reverse auction and has to pay, in addition to the cost of production, also the cost of transportation from the supply-locations to the demand-locations. They do not guarantee SBB.
\cite{GGP2007} generalized the above settings to a unified trade reduction procidure. 

\cite{M08} designs fixed price SBB double auction under the assumptions that the buyer's distribution dominates the seller's distribution or that there is  exponential distribution. Our result does not assume knowledge of the distribution of participating categories. Additionally, we also allow for general number of categories as opposed to two.

\citet{CKLT16,colini2017approximately} also presents SBB auctions. Their auctions target double-sided and combinatorial markets respectively. However, their goal is to maximize social welfare as opposed to our goal which is maximizing gain from trade\footnote{When optimizing gain from trade we optimize the difference between the total value of the sold items for the buyers and the total value of these items for the sellers. When optimizing social welfare in a market we optimize the sum of the buying agents' valuations plus the sum of the unsold items' value held by selling agents at the end of trade. Despite their conceptual similarity, the two objectives are rather different in approximation. In many cases the social welfare approximation is close to the optimal social welfare solution; however, the gain from trade approximation may not be within any constant factor of the optimal gain from trade.}.
Thus their mechanisms are not asymptotically-optimal for gain from trade. They also require a prior on the agents' valuations.

\cite{FFG18,FG18} presents a multilateral randomized market with buyers, mediators and sellers in the context of ad auctions. Their sellers are pre-associated with the mediators and they assume that mediators and buyers arrive over time in some uniform random order. Moreover, their trade matches are conducted in two stages: first the mediator trade with the buyer on behalf of his sellers and then the mediator transfers payments to his matched sellers. Our auction unites all three categories of buyer, seller and mediator actions into a single simultaneous trade step. 

Some works attempt to handle mechanisms with buyers and sellers interacting through an intermediary or trader such as \cite{FMMP10} and \cite{BEKT07}. However, their design reduces to a one-sided auction 
or a two-sided auction. 

The present work handles multiple categories of agents, but each agent is single-parametric. An orthogonal line of work \cite{SegalHalevi2018MUDA,segal2018mida,gerstgrasser2019multi,GE2017} remains with two agent categories (buyers and sellers), but aims to handle multi-parametric agents. 
Another orthogonal line of work gets around Myerson --Satterthwaite in a different way: it relaxes truthfulness but keeps the maximum GFT. See e.g. \citet{lubin2008ice}.

\section{Experiments}
\label{sec:experiments}
We evaluated the performance of our auctions using a simulation experiment similar to the one described by \citet{mcafee1992dominant}. 
For several values of $n$ between $2$ and $1000$, 
and several recipe vectors,
we constructed a market with $n\cdot r_g$ agents of each category $g$, such that the  potential number of procurement-sets is $n$.
The value of each trader was selected uniformly at random.
For each $n$,
we made $50,000$ runs 
and averaged the results.

We conducted three experiments.%
\footnote{
The code used for the experiments and the experiment results are available at  \url{https://github.com/erelsgl/auctions}.
}

\subsection{Two-sided market}

\begin{table*}
\begin{center}
\begin{tabular}
{||c|c||c|c|c||c|c||c|c||}
\hline 
 & & \multicolumn{3}{|c||}{McAfee} & \multicolumn{2}{|c||}{External Comp.} &  \multicolumn{2}{|c||}{Ascending Price}  
\\ 
\hline 
$n$ & $k$ & $k'$ & Total GFT & Market GFT &  $k'$ & GFT  & $k'$ & GFT  
\\ 
\hline 
\begin{filecontents*}{experiment_comparing_mcafee_to_sbb_concise.csv}
n,k,mcafeek,mcafeetgft,mcafeemgft,extcompk,extcompgft,ascpricek,ascpricegft,mcafeenhk,mcafeenhtgft,mcafeenhmgft
2,0.99,0.58,77.18,69.01,0.5,62.69,0.5,62.26,0.17,24.96,16.57
3,1.49,1.08,88.08,80.05,1,77.75,1,77.93,0.55,52.45,33.55
4,2,1.57,92.45,83.9,1.5,84.39,1.5,84.36,1.02,69.11,44.98
5,2.49,2.04,94.61,86.08,2,88.03,2,88.01,1.51,79.11,53.42
10,4.99,4.52,98.37,91.78,4.5,94.53,4.5,94.55,3.99,94.32,73.46
15,7.5,7.02,99.22,94.18,7,96.46,6.99,96.44,6.5,97.43,81.59
25,12.5,12.02,99.71,96.35,12,97.92,12.03,97.92,11.51,99.06,88.57
50,24.99,24.49,99.92,98.08,24.5,98.98,24.48,98.99,24,99.76,94.16
100,50,49.5,99.98,99.01,49.49,99.49,49.48,99.5,48.98,99.94,97.04
500,249.98,249.48,100,99.8,249.48,99.9,249.48,99.9,249.04,100,99.4
1000,500.02,499.52,100,99.9,499.44,99.95,499.52,99.95,498.99,100,99.7
\end{filecontents*}
\csvreader
[head to column names]{experiment_comparing_mcafee_to_sbb_concise.csv}{}%
{
\n & \k &
\mcafeek & \mcafeetgft & \mcafeemgft &
\extcompk & \extcompgft &
\ascpricek & \ascpricegft 
\\
}%
\\\hline
\end{tabular} 
\end{center}
\caption{
\label{tab:results-11}
McAfee's mechanism versus our SBB mechanisms with the recipe $\mathbf{r}=(1,1)$.
}
\end{table*}

The first experiment was designed to compare our SBB mechanisms with  McAfee's mechanism. Since McAfee's mechanism is designed for procurement-sets with 1 buyer and 1 seller, we executed both SBB mechanisms with $\mathbf{r}=(1,1)$.
Each buyer's value was selected from $[1,1000]$ and each seller's value was selected from $[-1,-1000]$.

In each run, we calculated $k$ (the number of deals in the optimal trade) and $OPT$ (the optimal gain-from-trade).
We found that the average value of $k$ was approximately $n/2$.

For each mechanism, we calculated 
$k'$ (the actual number of deals done by the mechanism) 
and the gain-from-trade.
For McAfee's mechanism, we calculated both the \emph{total gain-from-trade} (including both the traders' gain and the auctioneer's profit), and the \emph{market gain-from-trade} (including only the traders' gain, without the auctioneer's profit). 
In our SBB mechanisms, the auctioneer's profit is zero by design, so these two measures are identical.

The results are presented at Table \ref{tab:results-11}. 
Below are some highlights.
\begin{itemize}
\item The actual number of trades ($k'$) is very near $k-0.5$ for all auctions.
Note that the theoretic lower bound is $k-1$, i.e., the mechanism might lose at most one optimal deal, but in average it loses about $1/2$ optimal deal. This is consistent with the results reported by \citet{mcafee1992dominant},
that his mechanism attains the optimal GFT in about 50\% of the cases.

\item Accordingly, the actual GFT ratio is much higher than the theoretical lower bound. For example, when $n=10$ (and $k\approx 5$), the theoretical lower bound is $80\%$, but all auctions attain more than $90\%$ of the optimal GFT.
The loss in gain-from-trade drops below $1/1000$ already for $k\geq 250$.

\item The External Competition and Ascending Price mechanisms have almost the same performance. This is not surprising, since they are essentially different implementations of the same rule.

\item 
The \emph{total GFT} of McAfee's mechanism is higher than the GFT of our SBB mechanisms. 
However, the \emph{market GFT} of McAfee's mechanism is, for all $n\geq 5$, lower than the GFT of our SBB mechanisms. 
This indicates that, for all but the most trivial markets, a strongly-budget-balanced auction is better for the traders, who enjoy all the surplus rather than yielding some of it to the auctioneer.

We do not have a sufficient explanation for this; perhaps this is due to McAfee's heuristic of ``guessing'' a price of $(b_{k+1}+s_{k+1})/2$. In any case, the difference becomes negligible when $k\geq 50$.
\end{itemize}

\subsection{One agent per category}

\begin{table*}
\begin{center}
\begin{tabular}{||c|c||c|c||c|c||c|c||c|c||}
\hline 
 & & \multicolumn{2}{|c||}{$2+1$ categories} & \multicolumn{2}{|c||}{$4+1$ categories} &  \multicolumn{2}{|c||}{$8+1$ categories}  
&  \multicolumn{2}{|c||}{$16+1$ categories} 
\\ 
\hline 
$n$ & $k$ & $k'$ & GFT &  $k'$ & GFT  & $k'$ & GFT  & $k'$ & GFT 
\\ 
\hline 
\begin{filecontents*}{experiment_sbb_with_vectors_of_ones_concise.csv}
n,k,aack,aacgft,aadk,aadgft,aaek,aaegft,aafk,aafgft
2,1,0.53,70.1,0.54,74.69,0.55,76.96,0.55,77.91
3,1.5,1.02,82.28,1.03,84.55,1.03,85.83,1.04,86.29
4,2,1.51,87.6,1.52,89.42,1.52,90.18,1.53,90.55
5,2.5,2.01,90.63,2.02,91.88,2.02,92.62,2.02,92.85
10,5.01,4.52,95.83,4.51,96.41,4.51,96.74,4.52,96.89
15,7.51,7.02,97.3,7.01,97.71,7.01,97.9,7.02,98.05
25,12.52,12.02,98.42,12.02,98.68,12.02,98.81,12.02,98.86
50,25.02,24.52,99.23,24.53,99.36,24.52,99.42,24.51,99.45
100,50.02,49.52,99.62,49.53,99.68,49.51,99.71,49.53,99.73
500,250.23,249.73,99.92,249.67,99.94,249.67,99.94,249.62,99.95
1000,500.37,499.87,99.96,499.85,99.97,499.79,99.97,499.77,99.97
\end{filecontents*}
\csvreader
[head to column names]{experiment_sbb_with_vectors_of_ones_concise.csv}{}%
{
\n & \k &
\aack & \aacgft &
\aadk & \aadgft &
\aaek & \aaegft &
\aafk & \aafgft 
\\
}%
\\\hline
\end{tabular} 
\end{center}
\caption{
\label{tab:results-111}
The SBB ascending-prices auction where the recipe is a vector of ones, with  different numbers of categories.
}
\end{table*}

\begin{table*}
\begin{center}
\begin{tabular}{||c|c||c|c||c|c||c|c||c|c||}
\hline 
 & & \multicolumn{2}{|c||}{$\mathbf{r}=(1,2)$} & \multicolumn{2}{|c||}{$\mathbf{r}=(1,4)$} &  \multicolumn{2}{|c||}{$\mathbf{r}=(1,8)$}  
&  \multicolumn{2}{|c||}{$\mathbf{r}=(1,16)$} 
\\ 
\hline 
$n$ & $k$ & $k'$ & GFT &  $k'$ & GFT  & $k'$ & GFT  & $k'$ & GFT 
\\ 
\hline 
\begin{filecontents*}{experiment_sbb_with_two_categories_concise.csv}
n,k,abk,abgft,adk,adgft,ahk,ahgft,apk,apgft
2,1,0.6,74.44,0.67,81.04,0.7,84.5,0.72,86.01
3,1.5,1.11,84.82,1.17,87.76,1.2,89.17,1.21,89.6
4,2,1.61,89.16,1.67,91.28,1.7,92.21,1.72,92.74
5,2.51,2.12,91.69,2.17,93.24,2.21,93.94,2.21,94.29
10,5,4.62,96.07,4.68,96.75,4.71,97.07,4.73,97.27
15,7.49,7.12,97.44,7.18,97.86,7.22,98.08,7.23,98.19
25,12.5,12.13,98.48,12.17,98.73,12.21,98.85,12.22,98.92
50,25.01,24.63,99.25,24.67,99.37,24.7,99.44,24.72,99.46
100,49.98,49.6,99.62,49.67,99.69,49.69,99.72,49.7,99.73
500,249.99,249.61,99.93,249.58,99.94,249.63,99.94,249.64,99.95
1000,499.83,499.46,99.96,499.48,99.97,499.5,99.97,499.49,99.97
\end{filecontents*}
\csvreader
[head to column names]{experiment_sbb_with_two_categories_concise.csv}{}%
{
\n & \k &
\abk & \abgft &
\adk & \adgft &
\ahk & \ahgft &
\apk & \apgft 
\\
}%
\\\hline
\end{tabular} 
\end{center}
\caption{
\label{tab:results-12}
The SBB ascending-prices auction with two categories and different recipes.
}
\end{table*}

The second experiment considered markets with a recipe that is a vector of ones, and a varying number of categories.
Since the first experiment indicated that the performance of the external competition auction is very similar to that of the ascending prices auction, we did the second experiment only for the ascending prices auction.

We checked markets with one category of ``buyers'' (agents with a positive value) and $s$ categories of ``sellers'' (agents with a negative value), where $s\in\{2,4,8,16\}$. 
As in the previous experiment, the value of each seller was selected uniformly at random from $[-1,-1000]$.
The value of each buyer was selected uniformly at random from $[1,1000\cdot s]$, such that the average value in the entire market remained $0$.
As in the first experiment, the average value of $k$ was $\approx n/2$.%
\footnote{
Based on our empirical results, 
we believe that $k\approx n/2$ whenever 
the value of each agent is selected uniformly at random, such that 
the average value of all agents in the market is $0$.
However, proving this formally is an exercise in probability analysis which is beyond the scope of the present paper.
}

The results are presented at Table \ref{tab:results-111}. 
Both $k'$ and the GFT increase as the number of agents in each PS increases. This increase becomes negligible when $n$ grows. Still, it is reassuring to see that the performance does not become worse when the PS become larger.

\subsection{Two unbalanced categories}

The third experiment checked markets with two categories, where the recipe is $(1,s)$, for $s\in\{2,4,8,16\}$.
The value of each buyer was selected uniformly at random from $[1,1000\cdot s]$, such that the average value in the entire market remained $0$.
As in the first experiment, the average value of $k$ was $\approx n/2$.

The results are presented at Table \ref{tab:results-12}. As in the previous experiment, both $k'$ and the GFT increase with the number of agents in each PS, but the effect becomes negligible when $n$ grows.

We conclude that, in all scenarios that we checked,
our auctions perform significantly better that than their worst-case guarantees.

\section{Future Work}
\label{sec:future}
We are working on extending the SBB auctions in two directions.

\subsection{Multiple procurement-set recipes}
In a general market, several kinds of procurement-sets may be possible. Such a market can be represented by a finite set $R$, which contains one or more recipe-vectors of length $|G|$.

Some simple special cases are already handled by the current mechanisms. For example, consider a market with three categories and $R = \{(1,1,0),(1,0,1)\}$.
In this market, each deal requires one agent of category 1 and one agent of category either 2 or 3.
It is easy to see that this market is equivalent to a market with a single recipe $R = \{(1,1)\}$, since categories 2 and 3 can be united and treated as a single category. Both the calculation of the optimal trade and the calculation of the external competition work exactly as with the current mechanism.

Things become more interesting when the recipes are not equivalent, for example when $R = \{(1,1,0),(2,0,1)\}$.
For concreteness, denote the the three categories by seller, buyer and combiner, so that each deal requires either a seller+buyer, or two sellers and a combiner.
In this case, the optimal trade can no longer be computed by a simple greedy algorithm as in the case of a single recipe. 
To see this, consider the special case in which there are $n$ sellers, all of whom have a value of $0$. 
Then the problem is just to choose buyers and combiners. 
This is an instance of the \textsc{knapsack} problem, in which the weight of each buyer is $1$ and the weight of each combiner is $2$. It is known that greedy algorithms yield sub-optimal solutions to this problem. For example, in the following market:
\begin{itemize}
\item Buyers: $[99, 51]$
\item Combiner: $[100]$
\item Sellers: $[0,0,...]$
\end{itemize}
If there are two sellers, then the optimal trade requires the two buyers, so the greedy algorithm that selects  higher-valued agents first is not optimal (it selects only the combiner).
On the other hand, if there are three sellers, then the optimal trade requires the combiner and the highest-valued buyer, so the greedy algorithm that selects agents with higher value/weight ratio is not optimal (it selects the two buyers). 

Since the weights are integers, the optimal trade can be computed in pseudo-polynomial time using a dynamic programming algorithm. 
However, the ascending-price auction effectively induces a greedy selection of traders (since it induces only the highest-valued traders to remain), which might be sub-optimal. 
We believe that, if the weights are sufficiently small and the market is sufficiently large, then the greedy algorithm (and hence the ascending-prices auction) can still yield a good approximation of the optimal GFT. However, proving this in all cases, as well as studying other combinations of recipes, is a challenge for future work.

\subsection{Transaction costs}
In general, each procurement-set may have a different cost-of-transaction, depending on the geographic locations of the agents in the PS, as well as other factors. 
Such transaction costs make the computation of the optimal trade difficult, even before strategic considerations, and even when all transaction costs are common knowledge.

We have preliminary results showing that, 
without any restrictions on the transaction costs, 
there might be no mechanism that satisfies all the desirable properties of Theorem \ref{thm:1recipe-ones}. 
It is interesting to see whether such a mechanism can be found under some natural restrictions on the
transaction costs, such as the ones described by \citet{Chu2006Agent}.

\section*{Acknowledgments}
A preliminary version was accepted to the AAAI 2020 conference. We are grateful to three referees of AAAI 2020 for their helpful comments.

The following were added in the present version:
\begin{itemize}
\item A description and proof of correctness of the external-competition auction for general procurement-set recipes (subsections \ref{sub:ec-auction} and \ref{sub:ec-proof}).
\item Simulation experiments illustrating the performance of our mechanisms (Section \ref{sec:experiments}).
\end{itemize}

\fontsize{9.3pt}{10.3pt} \selectfont
\bibliographystyle{aaai}
\bibliography{AAAI-GonenR.9455}

\end{document}